\documentclass[conference]{ieeeconf}  

\IEEEoverridecommandlockouts

\usepackage{graphicx} 
\usepackage{amsmath} 
\usepackage{array}   
\usepackage{booktabs} 
\usepackage{caption} 
\usepackage{siunitx} 
\usepackage{epstopdf} 

\usepackage{cite}
\usepackage{amsmath,amssymb,amsfonts,mathtools,nccmath,bm, tikz, enumerate}
\usepackage{multicol}
\usepackage[nodisplayskipstretch]{setspace}
\newtheorem{definition}{Definition}[section] 
\newtheorem{assumption}{Assumption}[section]

\newtheorem{theorem}{Theorem}
\newtheorem{lemma}{Lemma}
\usepackage{algorithmic}
\usepackage{graphicx}
\usepackage{textcomp}
\usepackage{cuted}
\usepackage{float}
\usepackage{xcolor}
\def\BibTeX{{\rm B\kern-.05em{\sc i\kern-.025em b}\kern-.08em
    T\kern-.1667em\lower.7ex\hbox{E}\kern-.125emX}}
\allowdisplaybreaks

\begin{document}

\title{\textbf{Distributed Adaptive Consensus with Obstacle and Collision Avoidance for Networks of Heterogeneous Multi-Agent Systems} 
}

\author{Armel Koulong, \IEEEmembership{Student Member, IEEE} and Ali Pakniyat, \IEEEmembership{Member, IEEE}
\thanks{A. Koulong and A. Pakniyat are with the department of \mbox{Mechanical} Engineering, University of Alabama, Tuscaloosa, AL, USA (e-mails: akoulongzoyem@crimson.ua.edu, apakniyat@ua.edu). }
}

\maketitle
\thispagestyle{plain}
\pagestyle{plain}

\begin{abstract}
This paper presents a distributed adaptive control strategy for multi-agent systems with heterogeneous dynamics and collision avoidance. We propose an
adaptive control strategy designed to ensure leader-following formation consensus while effectively managing collision and obstacle avoidance using potential functions. By integrating neural network-based disturbance estimation and adaptive tuning laws, the proposed strategy ensures consensus and stability in leader-following formations under fixed topologies. 
\end{abstract}

\section{Introduction} \label{Introduction}

The cooperative control of multi-agent systems (MAS) is essential for applications such as unmanned aerial vehicle (UAV) coordination, autonomous vehicle platooning, and smart grids. In these systems, agents interact only with a limited number of other agents, creating a graph-theoretic structure that complicates achieving global consensus \cite{lewis2013,1431045}. This challenge is further amplified in heterogeneous uncertain MAS, where agents with different forms of high-order nonlinear dynamics are subject to external disturbances \cite{Khoo2009}. Moreover, adapting to unknown parameters in real-time further complicates the consensus process, as agents must continuously update their internal models based on limited interactions. 

Despite notable developments in the field, there are still significant gaps in addressing the simultaneous challenges of adaptive control for leader-following consensus with obstacle and collision avoidance in nonlinear systems. Approaches that deal with adaptive tracking control  \cite{lewis2013}, \cite{Hong2006} and consensus in switching topologies \cite{Qin2011} often do not fully address the physical constraints associated with collision avoidance. Similarly, adaptive tracking control methods \cite{Hong2006} have proven effective in systems with dynamic topologies and active leaders, yet their scope does not extend to handling heterogeneous agent dynamics or fixed topology networks, which are critical in real-world applications. Research into local interaction rules~\cite{Olfati2006} and conditions for formation control~\cite{Lin2005} have provided key insights into geometric constraints in specific systems, but these approaches generally overlook the interplay between consensus and physical safety constraints such as collision avoidance. Additionally, stability analyses of MAS under time-dependent communication links \cite{Moreau2005} highlight challenges related to fluctuating network conditions but do not address the need for real-time nonlinear estimation. Furthermore, work on target aggregation and state agreement under switching topologies \cite{Shi2009} explores consensus in nonlinear MAS, yet it only partially considers the complexities posed by physical constraints and heterogeneous dynamics.

Expanding on these foundations, this paper proposes an adaptive control strategy designed to ensure leader-following formation consensus while effectively managing collision and obstacle avoidance using potential functions. This strategy uniquely integrates solutions to these challenges, providing a robust framework for heterogeneous MAS with fixed topologies and external disturbances. The proposed method offers a unified approach, addressing the limitations of previous research by incorporating all the critical features simultaneously. 
The rest of the paper is structured as follows: 

Section \ref{ProblemFormulation} which is the problem formulation defines the MAS dynamics and leader-following framework. The methodology presented in Section \ref{Methodology} describes the proposed adaptive control protocol and neural network-based learning for managing nonlinear dynamics and disturbances. The main results in Section \ref{MainResult} proves the stability and convergence of the control laws. Section \ref{NUMERICALEXAMPLE} provides a numerical example that demonstrates the application of the control strategy in a multi-agent leader-follower formation scenario. Section \ref{CONCLUSION} summarizes the research contributions and potential future applications.

\textbf{Notations:} 
Throughout the paper, absolute value is denoted by \( | \cdot | \); the Euclidean norm of a vector by \( \| \cdot \| \); the Frobenius norm of a matrix by \( \| \cdot \|_F \); the trace of a matrix by \( \text{tr}\{ \cdot \} \); and the set of singular values of a matrix by \( \sigma(\cdot) \), with \( \bar{\sigma}(\cdot) \) and \( \underline{\sigma}(\cdot) \) representing the maximum and minimum singular values, respectively. Other notations will be introduced in the text at their first appearance.

\section{Problem Formulation} \label{ProblemFormulation}
We define the set of follower agents as \( \mathcal{N} = \{1,...,N\} \). The dynamics of the \(i^\text{th}\) follower agent is described using the nonlinear Brunovsky form as follows:
\begin{equation}\label{planta}
\begin{aligned}
\dot{x}_{i}^{1} &= x_{i}^{2}\\
\dot{x}_{i}^{2} &= x_{i}^{3} \\
&~\,\vdots \\
\dot{x}_{i}^{n} &= f_{i}(x_{i}) + u_{i} + w_{i}
\end{aligned}
\end{equation}
where \(k=\{1, 2, \ldots, n\}\), \(x_{i}^{k} \in \mathbb{R}\) is the \(k\)-th state of agent~\(i\), and \(x_i = [{x}_{i}^{1},{x}_{i}^{2},...,{x}_{i}^{n}] \in \mathbb{R}^n \) denoting its state vector, \(u_{i} \in \mathbb{R} \) represents the control input of agent \(i\), and \(w_{i} \in \mathbb{R} \) denotes the bounded unknown time-varying disturbance affecting agent \(i\). 

The unknown functions \( f_i(x) : \mathbb{R}^n \rightarrow \mathbb{R} \) are locally Lipschitz in \( x \) with \( f_i(0) = 0 \), for all $i \in \mathcal{N}$.
Although the governing dynamics \eqref{planta} of the agents are decoupled, their interactions become coupled through collision and obstacle avoidance, as discussed further in this~section.

We utilize the Kronecker product \cite{Wang2011} to collectively represent \eqref{planta} in the following global form:
\begin{equation}\label{plantag}
\begin{aligned}
\dot{x}^{1} &= x^{2}\\[-5pt]
&~\,\vdots \\
\dot{x}^{n} &= f(x) + u  + w 
\end{aligned}
\end{equation}
where $x^{k} = [x_{1}^{k}, x_{2}^{k},\ldots, x_{N}^{k}]^{\top} \in \mathbb{R}^N$ denotes the state vector of the follower agents, $u = [u_{1}, u_{2}, \ldots, u_{N}]^{\top}  \in \mathbb{R}^N$ denotes their control inputs, and $w = [w_{1}, w_{2}, \ldots, w_{N}]^{\top}  \in \mathbb{R}^N$ denotes unknown bounded disturbances. The unknown function $f(x) = [f_{1}(x_{1}), f_{2}(x_{2}), \ldots, f_{N}(x_{N})]^{\top}  \in \mathbb{R}^{n N}$ represents the dynamics of the entire $N$ follower agents.  

In this paper, we consider a leader-follower scenario where the leader agent, denoted by subscript $0$, generates a reference trajectory for follower agents, but this reference trajectory is a priori unknown to all follower agents, and the only information about the reference trajectory available to follower agents is the leader's state vector \(x_0 \equiv x_0(t) \in \mathbb{R}^n\) at the current time $t \in [t_0,\infty)$.

The time-varying dynamics of the leader agent is:
\begin{equation}\label{plantref}
\begin{aligned}
\dot{x}_{0}^{1}(t) &= x_{0}^{2}(t) \\[-5pt]
& ~\,\vdots \\
\dot{x}_{0}^{n}(t) &= f_{0}(x_{0},t)
\end{aligned}
\end{equation}
where \(x_{0}^{k} \in \mathbb{R}\) is the \(k\)-th component of the leader agent's state, and \(x_0 = [{x}_{0}^1,{x}_{0}^2,...,{x}_{0}^n]\) is its state vector. 

The unknown function \( f_0(x_{0},t) : [0,\infty) \times \mathbb{R}^n \rightarrow \mathbb{R} \) in the leader dynamics is considered to be 
piecewise continous in~\(t\) and locally Lipschitz in \(x_0\), with \( f_i(0,t) = 0 \) for all $t \ge 0$, and all \(x_0 \in \mathbb{R}^n\). %

In the consensus problem, it is desired (see, e.g., \cite{WeiRen}) that 
\begin{align} 
    \lim_{t \to \infty} [(x^{k}_{i}-\psi_{i}) - (x^{k}_{0}- \psi_{0})] &= 0 \label{lam} \\
    \lim_{t \to \infty} [(x^{k}_{i}-\psi_{i}) - (x^{k}_{j}- \psi_{j})] &= 0 \label{la}
\end{align}
for each $i \in \mathcal{N}$ and for all $j  \in \mathcal{N}$, $j \ne i$, where \(\psi_{i}\), \(\psi_{j}\) and \(\psi_{0}\) denote the desired offsets for agents \(i\), \(j\) and \(0\), respectively. However, such a strong level of connectivity which requires communication among all agents is often not feasible in networked control systems due to practical constraints such as limited bandwidth, packet loss, and delays (e.g., see \cite{hespanha2007networked}). 
 
In this paper, we seek a decentralized consensus policy where, in addition to its own state \( x_i \), each agent determines its input \( u_i \) based solely on the states of a limited number of other agents, which may or may not include the leader agent. Furthermore, each agent must also account for obstacle and collision avoidance. 

To address collision avoidance between agents as well as obstacle avoidance, we employ potential functions as an extension of \cite{Wang2011}. These potential functions guide the agents' motions, ensuring that they avoid both collisions with each other and obstacles in their environment.

For collision avoidance between agent $i$ and agent $j$ positions, we define $m_{ij}$ as follows:
\begin{equation}  \label{col1}
m_{ij} :=
    \begin{cases}
        \quad 0 &  \| x_i^1 - x_j^1 \| \ge \psi_{ij}  \\
        \frac{\chi}{\left|\left|x_i^1-x_j^1\right|\right|}   &  \| x_i^1 - x_j^1 \| < \psi_{ij}
    \end{cases}
\end{equation}
Similarly, for collision avoidance between agent $i$ and the leader agent $0$, we define $m_{i0}$ as follows:
\begin{equation}  \label{col2}
m_{i0} :=
    \begin{cases}
        \quad 0 &   \| x_i^1 - x_0^1 \| \ge \psi_{i0}  \\
        \frac{\chi}{\left|\left|x_i^1-x_0^1\right|\right|}   & 
        \| x_i^1 - x_0^1 \| < \psi_{i0}
    \end{cases}
\end{equation}
Moreover, for obstacle avoidance between agent $i$ and obstacle $\Omega$, we define $m_{i{\bar{b}}}$ as follows:
\begin{equation} \label{obs1}
m_{i{\bar{b}}} :=
    \begin{cases}
        \quad 0 &  R < ||x_i^1 - \Omega||  \\
        \Bigg[\frac{R^2 - ||x_i^1 - \Omega||^2}{\left|\left|x_i^1 - \Omega\right|\right|^2 - \partial^2}\Bigg]^2   & \partial < ||x_i^1 - \Omega \| \le R
    \end{cases}
\end{equation}
The formulation for obstacle avoidance between the leader agent \( 0 \) and obstacle $\Omega_{\bar{b}}$ is the same, with \( x_i^1 \) replaced by \( x_0^1 \).
In these expressions, \(\chi\) is a positive scalar adjusting the repulsive force strength, \(\psi_{ij}\) is the desired separation between agents \(i\) and \(j\), \(\psi_{i0}\) is the desired separation between agent \(i\) and the leader, \(R\) is the obstacle detection radius, and \(\partial\) is the obstacle radius. 

We define a communication topology \(\mathcal{G} = (\mathcal{V}, \mathcal{E}, A)\), where \(\mathcal{V} = \{v_1, v_2, \dots, v_N\}\) is the set of agents, and \(\mathcal{E} \subseteq \mathcal{V} \times \mathcal{V}\) is the set of edges. A graph is undirected if \((v_i, v_j) \in \mathcal{E} \iff (v_j, v_i) \in \mathcal{E}\); otherwise, it is directed if \((v_i, v_j) \in \mathcal{E}\) but \((v_j, v_i) \notin \mathcal{E}\). The weighted adjacency matrix \(A = [a_{ij}] \in \mathbb{R}^{N \times N}\) models the interaction strength between agents, where:
\[
  a_{ij} :=
  \begin{cases}
    w_{ij}  & \text{if } (i,j) \in \mathcal{E}, \\
    0 & \text{otherwise}.
  \end{cases}
\]
Here, \(w_{ij}\) represents the positive edge weight between agents \(i\) and \(j\). The in-degree matrix \(D = \text{diag}\{d_i\} \in \mathbb{R}^{N \times N}\) is diagonal, where \(d_i = \sum_{j=1}^{N} a_{ij}\) is the sum of the edge weights connected to node \(i\). The Laplacian matrix \(L = D - A\) is positive semi-definite, with one zero eigenvalue if the graph is connected. With a leader agent, the augmented graph is \(\bar{\mathcal{G}} = (\bar{\mathcal{V}}, \bar{\mathcal{E}})\), where \(\bar{\mathcal{V}} = \{v_0, v_1, \dots, v_N\}\), and \(B = \text{diag}\{b_{i}^{0}\} \in \mathbb{R}^{N \times N}\) denotes the leader's interaction, where:

\[
  b_{i}^{0} :=
  \begin{cases}
    w_{i0}  & \text{if } (i,0) \in \bar{\mathcal{E}}, \\
    0 & \text{otherwise}.
  \end{cases}
\]
Here, \(w_{i0}\) represents the positive weight between agent \(i\) and the leader \(0\). The augmented graph \(\bar{\mathcal{G}}\) includes a spanning tree rooted at the leader \(0\), facilitating agent coordination through state communication.

\begin{assumption} \label{assp1}
The adjacency matrix $A$ meets the following conditions:
\begin{enumerate}[(a)]
\item The augmented graph \( \bar{\mathcal{G}} \) contains a spanning tree with the leader as the root; this ensures that no clusters of agents are isolated from the leader. 
In other words, whenever $(i,0) \notin \bar{\mathcal{E}}$, there exists a sequence of nonzero elements of~$A$ of the form ${a}_{i i_2}, {a}_{i_2 i_3}, \cdots, {a}_{i_{l-1} i^{\prime}}$, for some $(i^{\prime},0) \in \bar{\mathcal{E}}$, with the sequence length $l$ being a finite integer.

\item If \( \Vert x_i - x_j \Vert \leq \Psi \) with \(\Psi \in \mathbb{R}_{>0}\) denoting a proximity threshold, then \( a_{ij} = 1 \); this ensures that agents within a $\Psi$-range are connected and can exchange information.
\end{enumerate}
\end{assumption}

\begin{assumption} \label{assp}
\hfill
\begin{enumerate} 
    \item The initial states of all follower agents and the leader agent are bounded, i.e., \(||x_i (t_0)|| \le X_n\) for all \(i \in \mathcal{N}\), and \(||x_0  (t_0)|| \le X_{n0}\).
    \item There exist continuous functions \( f_n(\cdot) \) and \( f_{n0}(\cdot) \) such that \( |f_i(x)| \leq |f_n(x)| \) for all \(i \in \mathcal{N}\), and \( |f_0(x_0,t)| \leq |f_{n0}(x_0)| \) for all \(x\) within the compact sets \(\varsigma_f = \{x \mid \|x\| \leq X_n\}\) and \(\varsigma_0 = \{x_0 \mid \|x_0\| \leq X_{n0}\}\).
    \item The input \(u(t)\) is bounded within any finite time interval, i.e., \(||u(t)|| \le u_n\) for all $t \in [t_0,T]$, $T\geq t_0$.
    \item The unknown disturbance \(w_i\) for each agent is  bounded, i.e., \(||w(t)|| \le w_n\) for all $t \in [t_0,\infty)$.
\end{enumerate}
\end{assumption}

\section{Methodology} \label{Methodology}
\subsection{Neural Network Learning Problem}
In order to account for uncertainties in the dynamics, we use a distributed two-layer 
linear-in-parameters (LIP) neural network \cite{yesildirak1995neural}, enabling each agent to locally update its model. In this framework, each agent independently uses a neural network to model the nonlinear dynamics of itself and the agents connected to it in real-time. This also yields to a reduced number of required neurons in comparison to centralized approaches. 
Accordingly, we model the functions as:
\begin{align} 
\label{orige}
f_{i}(x_{i}) &= \theta_{i}^{\top}\phi_{i}(x_{i}) + \varepsilon_i, \\ 
\label{orige1}
f_{0}(x_{0},t) &= \theta_{0}^{\top}\phi_{0}(x_{0},t) + \varepsilon_{0}, \\
\label{orige2}
w_{i}(t) &= \theta_{iw}^{\top}\phi_{iw}(t) + \varepsilon_{iw}.
\end{align}
where \( \phi_{i}(x_{i}) \), \( \phi_{0}(x_{0},t) \), and \( \phi_{iw}(t) \) are fixed basis functions, with the corresponding weight vectors \( \theta_{i} \), \( \theta_{0} \), and \( \theta_{iw} \), updated in real-time as new data is received. The approximation errors are \( \varepsilon_i \), \( \varepsilon_0 \), and \( \varepsilon_{iw} \).

The neural network approximations for \( f_{i}(x_{i}) \), \( f_{0}(x_{0},t) \), and \( w_{i}(t) \) are therefore
\(
\hat{f}_{i}(x_{i}) = \hat{\theta}_{i}^{\top}\phi_{i}(x_{i}), \quad \hat{f}_{0}(x_{0},t) = \hat{\theta}_{0}^{\top}\phi_{0}(x_{0},t), \quad \hat{w}_{i}(t) = \hat{\theta}_{iw}^{\top}\phi_{iw}(t)
\), where the NN weights \( \hat{\theta}_{i} \), \( \hat{\theta}_{0} \), and \(\hat{\theta}_{iw}\) are computed locally. 
For the brevity of notation, the basis functions and errors are denoted by
\(
\theta = \text{diag}(\theta_1, \ldots, \theta_N), \quad 
\phi(x) = [\phi_1^{\top}(x_1), \ldots, \phi_N^{\top}(x_N)]^{\top}, \quad 
\varepsilon = [\varepsilon_1, \ldots, \varepsilon_N]^{\top}.
\)
Other parameters follow a similar pattern.
The global nonlinearities are expressed as:
\begin{equation}\label{fxes}
     f(x) = {\theta}^{\top} \phi(x) + \varepsilon 
\end{equation}
\begin{equation}\label{f0es}
     f_0(x,t) = {\theta}_0^{\top} \phi_0(x,t) + \varepsilon_0
\end{equation}
\begin{equation}\label{fwes}
     w(t) = {\theta}_w^{\top} \phi_w(t) + \varepsilon_w 
\end{equation}
with approximations:
\begin{equation}\label{fxesh}
      \hat{f}(x) = \hat{\theta}^{\top} \phi(x)  
\end{equation}
\begin{equation}\label{f0esh}
      \hat{f}_0(x,t) = \hat{\theta}_0^{\top} \phi_0(x,t) 
\end{equation}
\begin{equation}\label{fwesh}
      \hat{w}(t) = \hat{\theta}_w^{\top} \phi_w(t)
\end{equation}

The NN weight errors are:
\( 
\tilde{\theta} = \theta - \hat{\theta}\), \(  
\tilde{\theta}_{0} = \theta_{0} - \hat{\theta}_{0}\), \(  
\tilde{\theta}_{w} = \theta_{w} - \hat{\theta}_{w}.
\) 

The \( k^\text{th} \) order tracking error is defined as \(\delta_{i}^k = (x_{i}^k - \psi_{i}) - (x_{0}^k - \psi_{0})\), and globally as \(\delta^{k} = [\delta_{1}^k, \delta_{2}^k, \ldots, \delta_{N}^k]^{\top} \in \mathbb{R}^N\).

\begin{definition}
The tracking error \(\delta^{k}\) is \textit{cooperatively uniformly ultimately bounded} (CUUB) \cite{lewis2013} if there exists a compact set \(\varsigma^k \subset \mathbb{R}^N\) with \(\{0\} \subset \varsigma^k\) such that for any initial condition \(\delta^k(t_0) \in \varsigma^k\), there exists \(B^k \in \mathbb{R}_{>0}\) and \(T_k \in \mathbb{R}_{>0}\) which yield \(||\delta^k(t)|| \le B^k\) for all \(t \ge t_0 + T_k\). \end{definition}

\begin{assumption} \label{assumption1}
The basis functions \( \phi_{i}(x_{i}) \), \( \phi_{0}(x_{0}, t) \), \( \phi_{iw}(t) \), the NN weights \( \theta_{i} \), \( \theta_{0} \), \( \theta_{iw} \), and the approximation errors \( \varepsilon \), \( \varepsilon_{0} \), \( \varepsilon_{w} \) are bounded by specified positive constants. 
\end{assumption}

Let \( \phi_{in} = \max_{x_i \in \varsigma} \|\phi_i(x_i)\| \), \( \phi_{0n} = \max_{x_0 \in \varsigma, t \geq 0} \|\phi_0(x_0, t)\| \), and \( \phi_{iwn} = \max_{t \geq 0} \|\phi_{iw}(t)\| \). Based on 
Assumption~\ref{assumption1},
there exist positive numbers \( \Phi_n \), \( \Theta_n \), and \( \varepsilon_n \); \( \Phi_{n0} \), \( \Theta_{n0} \), and \( \varepsilon_{n0} \); \( \Phi_{nw} \), \( \Theta_{nw} \), and \( \varepsilon_{nw} \), such that:
\(
\|\phi_i(x_i)\| \leq \Phi_n, \quad \|\phi_0(x_0, t)\| \leq \Phi_{n0}, \quad \|\phi_{iw}(t)\| \leq \Phi_{nw},
\)
\(
\|\theta_i\|_F \leq \Theta_n, \quad \|\theta_0\|_F \leq \Theta_{n0}, \quad \|\theta_{iw}\|_F \leq \Theta_{nw},
\)
\(
\|\varepsilon\| \leq \varepsilon_n, \quad \|\varepsilon_0\| \leq \varepsilon_{n0}, \quad \|\varepsilon_w\| \leq \varepsilon_{nw}.
\)

\subsection{Local Synchronization Error}
Given the limited information available to the $i^\text{th}$ agent, 
we define its $k^\text{th}$ order weighted synchronization error as:
\begin{multline}{\label{errordyn}}
e_i^k = - \nu_1 \sum_{j=1}^{N} a_{ij} \bigg[(x_{i}^{k} - \psi_{i}) - (x_{j}^{k} - \psi_{j})\bigg] \\
- \nu_2 b_{i}^{0} \bigg[(x_{i}^{k} - \psi_{i}) - (x_{0}^{k} - \psi_{0})\bigg]
\end{multline}
where the parameters \(\nu_1 \text{ and } \nu_2, \) are positive scalar gains that adjust the influence between agents and the leader, providing flexibility in varying degrees of influence across the network.

Using the simplifying notations:
\(\bar{x}^k_i := x_{i}^{k} - \psi_{i}\),
\(\bar{x}^k_j :=  x_{j}^{k} - \psi_{j}\) and
\(\bar{x}^k_0 :=  x_{0}^{k} - \psi_{0}\),
we rewrite equation \eqref{errordyn} as:
\begin{multline}{\label{errordynu}}
e_i^k = - \nu_{1} \sum_{j=1}^{N} a_{ij} \bigg[\bar{x}^k_i - \bar{x}^k_j\bigg]
- \nu_{2} b_{i}^{0} \bigg[\bar{x}^k_i - \bar{x}^k_0\bigg]
\end{multline}
Here,  \( - \nu_{1} \sum_{j=1}^{N} a_{ij} (\bar{x}^k_i - \bar{x}^k_j) \) ensures consensus among agents by penalizing the difference between the state of agent \(i\) and its neighbors, while \( - \nu_{2} b_{i}^{0} (\bar{x}^k_i - \bar{x}^k_0) \) couples the agent's state to the leader state by penalizing their difference. Defining,
\(
{e}^k = [{e}^k_{1}, {e}^k_{2}, {e}^k_{3}, \dots, {e}^k_{N}]^\top,
\)
\(
\bar{\underline{x}}_{0}^k = [\bar{x}_{0}^k, \dots, \bar{x}_{0}^k]^\top,
\)
\(
\bar{x}^{k} = [\bar{x}_{1}^{k},\dots, \bar{x}_{N}^{k}]^{\top}
\),
we reformulate \eqref{errordynu} in global form as:
\begin{align}
e^k = -(\nu_1 L + \nu_2 B)(\bar{x}^k - \bar{\underline{x}}_{0}^k)
\end{align}

\begin{lemma} \label{rem:lbnon}
Under Assumption \ref{assp1}, 
the matrix \(\nu_1 L + \nu_2 B\) is nonsingular.
\end{lemma}

\begin{proof} Let \(J\) denote the set of strictly diagonally dominant nodes of \(\bm{\pounds}:=\nu_1 L + \nu_2 B\), i.e,
\[
J = \{ i : |\pounds_{ii}| > \sum_{j \neq i} |\pounds_{ij}| \}.
\]

If \(J = \mathcal{N}\), i.e., the matrix \(\bm{\pounds}\) is strictly diagonally dominant, then  \(\det(\bm{\pounds}) \neq 0\) is directly obtained by the application of the Gershgorin circle theorem \cite{qu2009cooperative}. 

For the case where \(J \subset \mathcal{N} \) is a strict subset of agents, we note from $\bm{\pounds}:=\nu_1 L + \nu_2 B$ and $L:=D-A$ that \mbox{$\bm{\pounds} = \nu_1 D - \nu_1 A +\nu_2 B$}. Since $B := \text{diag}\{b_{i}^{0}\}$ and \mbox{$D := \text{diag}\{d_i\}$} are diagonal matrices, the on-diagonal elements of $\bm{\pounds}$ are $\bm{\pounds}_{ii} = -\nu_1 d_{ii} -\nu_2 b_{i}^{0}$ and the off-diagonal elements are $\bm{\pounds}_{ij} = -\nu_1 a_{ij}$, $i \neq j$. 

These relations, together with Assumption~\ref{assp1}, yield that: (i) \(J \neq \emptyset\) and (ii) for each agent \(i \notin J\), there exists a sequence of nonzero elements of $\bm{\pounds}$ in the form ${\pounds}_{i_1 i_2}, {\pounds}_{i_2 i_3}, \cdots, {\pounds}_{i_{s-1} i_s}$, for some $i_s \in J$, which directly correspond to the sequence of non-zero elements $ a_{i_1 i_2}, a_{i_2 i_3}, \cdots, a_{i_{s-1} i_s}$ in Assumption~\ref{assp1}. Thus, by invoking \cite[Theorem]{shivakumar1974sufficient}, we establish that $\bm{\pounds}$ is non-singular.
\end{proof}

\subsection{Local Error Dynamics}
The error dynamics are derived by differentiating the synchronization error from \eqref{errordynu}:
\begin{align}
\label{plantas}
\dot{e}^{k} &= e^{k+1}, \quad k=1,...,n-1\notag \\
\dot{e}^{k} &= -(\nu_{1} L + \nu_{2} B)(\dot{\bar{x}}^n - \dot{\bar{\underline{x}}}_0^n), \quad k=n
\end{align}
which, in the expanded form, is written as:
\begin{align} \label{explantas}
\dot{e}^1 &= e^2 = -(\nu_1 L+\nu_2 B)(\bar{x}^2 - \bar{\underline{x}}_{0}^2) \notag \\
\dot{e}^2 &= e^3 = -(\nu_1 L+\nu_2 B)(\bar{x}^3 - \bar{\underline{x}}_{0}^3) \notag \\
&\vdots \notag \\
\dot{e}^n &= -(\nu_1 L+\nu_2 B)(f(\bar{x}) + u + w - f_0(\bar{x}_0,t))
\end{align}

\subsection{Weighted Stability Error}
The weighted stability error \(r_i\) is defined as a linear combination of error terms for each follower agent \(i\):
\begin{equation}
 r_i = \lambda_{1}e_{i }^1 + \lambda_{2}e_{i }^2 + \dots + \lambda_{n-1}e_{i}^{n-1} + e_{i}^n   
\end{equation}
where the design parameters \(\lambda_j\) (for \( j = 1, \ldots, n-1 \)) are chosen such that the characteristic polynomial:
\begin{equation} \label{eq:hurwitz_poly}
s^{n-1} + \lambda_{n-1} s^{n-2} + \cdots + \lambda_1
\end{equation}
is Hurwitz. This ensures that all roots have negative real parts, leading to the stability of the associated linear system. The selection of \(\lambda_j\) can alternatively be made by setting
$s^{n-1} + \lambda_{n-1} s^{n-2} + \cdots + \lambda_1 = \prod_{j=1}^{n-1} (s - \xi_j)
$, and selecting \(\xi_j\)'s to be positive real numbers. The global form of the representation of the weighted stability error is:
\begin{equation} \label{slidem}
r = \lambda_{1}e^{1} + \lambda_{2}e^{2} + \dots + \lambda_{n-1}e^{n-1} + e^{n} 
\end{equation}
By differentiating \eqref{slidem} with respect to time, the weighted stability error dynamics are:
\begin{multline}
\label{slidemd}
\hspace{-9pt}\dot{r} = \lambda_{1}\dot{e}^{1} + \lambda_{2}\dot{e}^{2} + \dots + \lambda_{n-1}\dot{e}^{n-1} + \dot{e}^{n} \\
 = \lambda_{1}{e}^{2} + \lambda_{2}{e}^{3} + \dots + \lambda_{n-1}{e}^{n} - (\nu_{1}L + \nu_{2}B)(\dot{\bar{x}}^n - \dot{\bar{x}}_0^n)
\end{multline}
which can be written as 
\begin{multline} \label{errdyn}
\dot{r} = \rho - (\nu_{1}L + \nu_{2}B)(f(x) + u + w - f_0(x_0,t)),
\end{multline}
where \( \rho = \lambda_{1}{e}^{2} + \lambda_{2}{e}^{3} + \lambda_{3}{e}^{4} + \dots + \lambda_{n-1}{e}^{n} = E_{2}\bar{\lambda} \) with \( \bar{\lambda} = [\lambda_1, \lambda_2,...,\lambda_{n-1}]^{\top} \) and \( E_2=[e^2, e^3,...,e^n]^{\top} \). 

Representing $E_2$ in matrix form yield: 
\[
E_2 = E_1 \bigtriangleup^{\top} + rl^{\top}
\]
where
\(
E_1 = [e^1,e^2,...,e^{n-1}] \in \mathbb{R}^{N\times(n-1)}
\), 
\(
E_2 = \dot{E}_1 = [e^2,e^3,...,e^n] \in \mathbb{R}^{N\times(n-1)}
\), 
\mbox{\(
l = [0,0,...,0,1]^{\top} \in \mathbb{R}^{(n-1)}
\)},
\[
\Delta = \begin{bmatrix}
0 & 1 & 0 & \cdots & 0 \\
0 & 0 & 1 & \cdots & 0 \\
\vdots & \vdots & \vdots & \ddots & \vdots \\
0 & 0 & 0 & \cdots & 1 \\
-\lambda_1 & -\lambda_2 & -\lambda_3 & \cdots & -\lambda_{n-1} \\
\end{bmatrix} \in \mathbb{R}^{(n-1) \times (n-1)}
\] 

As stated in \cite{lewis2013}, since $\bigtriangleup$ is Hurwitz,  there exists a positive-definite matrix \( P_1 \) such that the following condition holds for any positive number \( \bar{\alpha} \):
\begin{equation} \label{huwitx}
    \bigtriangleup^{\top} P_1 + P_1 \bigtriangleup = - \bar{\alpha}I
\end{equation}
where \( I \in \mathbb{R}^{(n-1)\times(n-1)} \).

\begin{lemma}
If \( r_i(t) \) is ultimately bounded, then \( e_i(t) \) is ultimately bounded. 
\end{lemma}
\begin{proof}
    See \cite[Lemma 10.3]{lewis2013}. 
\end{proof}

\begin{lemma}[Graph Lyapunov Equation] \label{lem:PQpositivedefinite}
Define 
\begin{equation}
    q \equiv [q_1,\cdots,q_N]^{\top} := (\nu_{1}L + \nu_{2}B)^{-1} \underline{1} ,
\end{equation}
\begin{equation}
    P \equiv \text{diag}\{p_i\}_{i \in \mathcal{N}} := \text{diag}\{1/q_i\}_{i \in \mathcal{N}},
\end{equation}
\begin{equation} \label{pdefinite}
    Q:=P(\nu_1L + \nu_2B) + (\nu_1L + \nu_2B)^{\top} P ,
\end{equation}
where \( \underline{1} = [1,...,1]^{\top} \in \mathbb{R}^N \). 
then the matrices \( P \) and \( Q \) are positive definite.
\end{lemma}
\begin{proof}
      The proof is similar to \cite[Lemma 10.1]{lewis2013} by noting that \(\nu_1 L + \nu_2 B\) is nonsingular by Lemma~\ref{rem:lbnon}.
\end{proof}


\section{Main Result} \label{MainResult}

\begin{theorem} 
\label{thm:MainResult}
Under the Assumptions \ref{assp1},  \ref{assp} and \ref{assumption1}, 
there exist a control law and NN tuning laws that guarantee synchronization, stability, and collision/obstacle avoidance in the distributed multi-agent with the follower dynamics  \eqref{planta} and leader dynamics \eqref{plantref}. 
\end{theorem}

\begin{proof}
    This theorem will be proved by construction in Section~\ref{sec:Proof}, i.e., by providing a control law and NN tuning laws. 
\end{proof}

Before proving Theorem~\ref{thm:MainResult}, we present a control strategy together with adaptive tuning strategy in Sections~\ref{control_law_defn} and~\ref{NN_Laws_Defn}, respectively. We then demonstrate that
applying the  distributed control law \eqref{controllaw} along with the NN tuning laws \eqref{adap1}, \eqref{adapt2}, \eqref{adapt3} guarantees the following:
\begin{enumerate}
    \item \textbf{Synchronization:} All follower agents achieve cooperative formation consensus tracking with the leader, such that the tracking errors \( \delta^1, \ldots, \delta^n \) are cooperatively uniformly ultimately bounded (CUUB). This implies that all agents 
    synchronize with the leader while maintaining bounded residual errors.
    \item \textbf{Stability:} The closed-loop system is globally bounded and remains stable for all \( t \geq 0 \).
    \item \textbf{Collision/Obstacle Avoidance:} Agents successfully avoid local agents/obstacles while maintaining synchronization with the leader and preserving system stability.
\end{enumerate}

\subsection{Proposed Distributed Control Law}\label{control_law_defn}
The proposed distributed control law is presented for each agent \(i\) as:
\[
u_i = u_i^d - u_i^c - u_i^0 
\]
where:
\begin{align}
    \hspace{-6pt}u_i^d &= \frac{\rho}{{d}_{i}+b_{i}^{0}} - \hat{\theta}_i^{\top}{\phi}_i - \hat{\theta}_{iw}^{\top}{\phi_{iw}} +\hat{\theta}_0^{\top}{\phi_0} + r_i - {c}_i E_{i0}^\top, \\
    \hspace{-6pt}u_i^0 &= \Gamma^0_b \sum_{b=1}^{\xi} m_{i\bar{b}}, \\
    \hspace{-6pt}u_i^c &= \Gamma^1_{ij} \sum_{j=1}^{N} m_{ij} +  \Gamma^2_{i0} \sum_{j=1}^{N} m_{i0}.
\end{align}
This yields the global form of the control law \( U \) as:
\begin{equation}\label{controllaw}
    U = U^d - U^0 - U^c
\end{equation}
where:
\begin{align}
    U^d &= {(D+B)}^{-1}{\rho} - \hat{\theta}^{\top}{\phi} - \hat{\theta}_{w}^{\top}{\phi_{w}} + \hat{\theta}_0^{\top}{\phi_0} + r - {c} E_0^\top, \\
    U^0 &= \Gamma^0 M^0, \\
    U^c &= \Gamma^1 M^c_I + \Gamma^2 M^c_0.
\end{align}
\( \Gamma^0 \) and \( \Gamma^1 \) are gain matrices that modulate the influence of obstacle avoidance (\(U^0\)) and collision avoidance (\(U^c\)). The overall control law is:
\begin{multline}
    U = {(D+B)}^{-1}{\rho} - \hat{\theta}^{\top}{\phi} - \hat{\theta}_{w}^{\top}{\phi_{w}} +\hat{\theta}_0^{\top}{\phi_0} + r - {c} E_0^\top\\ - \Gamma^0 M^0 - \Gamma^1 M^c_I - \Gamma^2 M^c_0
\end{multline}
Here, \( U = [u_1, u_2, \cdots, u_N] \) is the control input vector for all agents, and \((D+B)^{-1}\) is the inverse of the diagonal matrix \( \text{diag}[({d}_{1}+b_{10}), ({d}_{2}+b_{20}), \dots, ({d}_{N}+b_{N0})] \). The global neural network weight matrices \( \hat{\theta}, \hat{\theta}_{w}, \hat{\theta}_0 \) correspond to the agents' dynamics, disturbances, and leader's dynamics, while the global basis function matrices \( \phi, \phi_w, \phi_0 \) relate to these respective dynamics. \( {c}_i = [k^1, k^2, \cdots, k^n] \in \mathbb{R}^{1 \times n} \) is a gain matrix and  
 \( E_{i0}^\top =  
\begin{bmatrix}
    \bar{x}^1_{i} - \bar{x}^1_{0}  \\
    \bar{x}^2_{i} - \bar{x}^2_{0}  \\
    \vdots \\
    \bar{x}^n_{i} - \bar{x}^n_{0}
\end{bmatrix}
\in \mathbb{R}^{n \times 1} \) is the relative error between the follower and leader. In addition, \( {c} \) and \(E_0\) are the global forms of \({c}_i\) and \( E_{i0}\), respectively.

\subsection{Proposed NN Local Tuning Laws}\label{NN_Laws_Defn}
The adaptive tuning laws are derived from a Lyapunov function that captures both control design at each agent and the communication graph topology.

The NN parameter adaptive tuning laws are:
\begin{equation}
\dot{\hat{\theta}}_i = - F_i \bigg[\phi_i r_i p_i(d_i + b_{i}^{0}) + \kappa_i\hat{\theta}_i \bigg],
\end{equation}

\begin{equation} 
\dot{\hat{\theta}}_0 = F_{i0}\bigg[\phi_0 r_i p_i(d_i + b_{i}^{0}) - \kappa_0\hat{\theta}_0 \bigg],
\end{equation}

\begin{equation}
\dot{\hat{\theta}}_{iw} = - F_{iw} \bigg[\phi_{iw} r_i p_i(d_i + b_{i}^{0}) + \kappa_{iw}\hat{\theta}_{iw} \bigg].
\end{equation}
and in the global context, the parameter estimate dynamics are:
\begin{equation} \label{adap1}
\dot{\hat{\theta}} = - F \bigg[\phi r^{\top} P(D + B) + \kappa\hat{\theta} \bigg],
\end{equation}

\begin{equation} \label{adapt2}
\dot{\hat{\theta}}_0 = F_0\bigg[\phi_0 r^{\top} P(D + B) - \kappa_0\hat{\theta}_0 \bigg],
\end{equation}

\begin{equation}  \label{adapt3}
\dot{\hat{\theta}}_{w} = - F_{w} \bigg[\phi_{w} r^{\top} P(D + B) + \kappa_{w}\hat{\theta}_{w} \bigg].
\end{equation}
Here, \(F_i\), \(F_{i0}\), and \(F_{iw}\) are positive definite matrices structured as \(F = \text{diag}(F_1,F_2,...,F_N)\), \(F_0 = \text{diag}(F_0,F_0,...,F_0)\), and \(F_w = \text{diag}(F_{1w},F_{2w},...,F_{Nw})\), respectively. The matrix \(P\) is defined as in \eqref{pdefinite}, and \(\kappa\), \(\kappa_{0}\), and \(\kappa_{w}\) are positive scalar tuning gains.

\subsection{Proof of Theorem~\ref{thm:MainResult}}
\label{sec:Proof}
\begin{proof}
\textbf{Part 1:} We consider the Lyapunov function:
\begin{equation}\label{allV}
    V = V_1 + V_2 + V_3 + V_4 + V_5 
\end{equation}
with each term given as \( V_1 = \frac{1}{2}r^{\top}Pr \), \(V_2 = \frac{1}{2}\text{tr}\{\tilde{\theta}^{\top}F^{-1}\tilde{\theta}\}  \), \(V_3 = \frac{1}{2}\text{tr}\{\tilde{\theta}_0^{\top}F_0^{-1}\tilde{\theta}_0\}    \), \(V_4 = \frac{1}{2}\text{tr}\{\tilde{\theta}_w^{\top}F_w^{-1}\tilde{\theta}_w\}   \), and \mbox{\(V_5 = \frac{1}{2}\text{tr}\{E_1 P_1(E_1)^{\top}\} \)}.
Taking the derivative of \( V_1 \) and invoking the property \( a^{\top} b = \text{tr}\{b \, a^{\top}\} \), we obtain:
\begin{multline}
    \dot{V}_1 = -  \frac{1}{2}r^{\top} Q r  - \text{tr}\{\tilde{\theta}^{\top}{\phi}r^{\top} P (D + B)\} + 
    \text{tr}\{\tilde{\theta}^{\top}{\phi} r^{\top} PA\} \\ -
    \text{tr}\{\tilde{\theta}_w^{\top}{\phi}_w r^{\top} P(D+B)\}   + r^{\top} PA(D + B)^{-1} \rho \\ + \text{tr}\{ \tilde{\theta}_w^{\top}{\phi}_w r^{\top} PA \} + \text{tr}\{ \tilde{\theta}_0^{\top}{\phi}_0 r^{\top} P(D+B) \} + \frac{1}{2}cr^{\top} Q E_0^\top \\ - \text{tr} \{ \tilde{\theta}_0^{\top}{\phi}_0 r^{\top} PA \}  - r^{\top} P(\nu_1L + \nu_2B)(\varepsilon + \epsilon_w - \epsilon_0) \\ + r^{\top} P(\nu_{1}L + \nu_{2}B)(\Gamma^0 M^0 + \Gamma^1 M^c_I + \Gamma^2 M^c_0).
\end{multline}
The derivatives of \( V_2 \), \( V_3 \), \( V_4 \), and \( V_5 \) are, respectively,
\(
    \dot{V}_2 = -\text{tr}\{ \tilde{\theta}^{\top} F^{-1} \dot{\hat{\theta}} \},  
    \dot{V}_3 = -\text{tr}\{ \tilde{\theta}_0^{\top} F_0^{-1} \dot{\hat{\theta}}_0 \},  
    \dot{V}_4 = -\text{tr}\{ \tilde{\theta}_w^{\top} F_w^{-1} \dot{\hat{\theta}}_w \}
\) \text{and}
\begin{equation} \label{V5}
     \dot{V}_5 \le -\frac{\beta}{2} || E_1 ||_F^2 + \bar{\sigma}(P_1)||l||||r||||E_1||_F.
\end{equation}

Hence, \( \dot{V} = \dot{V}_1 + \dot{V}_2 + \dot{V}_3 + \dot{V}_4 + \dot{V}_5 \) satisfies:
\begin{multline}
    \dot{V} \le -\bigg [ \frac{1}{2}  \underline{\sigma}(Q) - \frac{\bar{\sigma}(P)\bar{\sigma}(A)}{\underline{\sigma}(D+B)}||\bar{\lambda}|| \bigg ] ||r||^2 
    \\
    + \bigg[ \frac{\bar{\sigma}(P)\bar{\sigma}(A)}{\underline{\sigma}(D+B)}||\bigtriangleup||_F||\bar{\lambda}||  + \bar{\sigma}(P_1) \bigg] ||r||||E_1||_F \\
    + \bar{\sigma}(P)\bar{\sigma}(\nu_1L + \nu_2B)T_M||r||
    - \kappa ||\tilde{\theta}||_F^2 -  \kappa_0 ||\tilde{\theta}_0||_{F_0}^2 
    \\ - \kappa_w ||\tilde{\theta}_w||_{F_w}^2 + \Phi_n \bar{\sigma}(P)\bar{\sigma}(A)||\tilde{\theta}||_F||r|| 
    \\ + \Phi_{nw} \bar{\sigma}(P)\bar{\sigma}(A)||\tilde{\theta}_w||_{F_w}||r|| + \Phi_{n0} \bar{\sigma}(P)\bar{\sigma}(A)||\tilde{\theta}_0||_{F_0}||r|| \\ +  \bar{\sigma}(P)\bar{\sigma}(\nu_1L + \nu_2B)T_N||r|| -\frac{\beta}{2} || E_1 ||_F^2 + \kappa \Theta_n||\tilde{\theta}||_F \\ + \kappa_w \Theta_{nw}||\tilde{\theta}_w||_{F_w} 
    + \kappa_0 \Theta_{n0}||\tilde{\theta}_0||_{F_0} + \frac{1}{2} c E_0 \underline{\sigma}(Q)||r||
\end{multline}
which can be written as
\begin{equation} \label{modv}
    \dot{V} \le -z^{\top}Kz + \omega^{\top}z = -V_z(z)
\end{equation}
with
\begin{equation*}
K = \begin{bmatrix}
    \frac{\beta}{2} & 0 & 0 & 0 & g \\
    0 & \kappa & 0 &  0 & \gamma_1 \\
    0 & 0 & \kappa_w & 0 & \gamma_2 \\
    0 & 0 & 0 & \kappa_0 & \gamma_3 \\
    g & \gamma_1 & \gamma_2 & \gamma_3 & \mu_1 \\
    \end{bmatrix},
\end{equation*}
\begin{align*}
    \omega &= \bigg [ 0,\kappa \Theta_n,\kappa_w \Theta_{nw}, \kappa_0 \Theta_{n0}, \Lambda \bigg]^{\top}, \\
     \Lambda &= \bar{\sigma}(P)\bar{\sigma}(\nu_1L + \nu_2B)(T_M+T_N) + \mu_2 , \\
   z &= \bigg[ ||E_1||_F, ||\tilde{\theta}||_F, ||\tilde{\theta}_w||_{F_w}, ||\tilde{\theta}_0||_{F_0}, ||r|| \bigg]^{\top} , \\
    h &= \frac{\bar{\sigma}(P)\bar{\sigma}(A)}{\underline{\sigma}(D+B)}||\bar{\lambda}|| , \\
    \gamma_1 &=  -\frac{1}{2}\Phi_n\bar{\sigma}(P)\bar{\sigma}(A), \\ 
    \gamma_2  &=  -\frac{1}{2}\Phi_{nw}\bar{\sigma}(P)\bar{\sigma}(A), \\
    \gamma_3 &=  -\frac{1}{2}\Phi_{n0}\bar{\sigma}(P)\bar{\sigma}(A), \\ 
    \mu_1  &=  \frac{1}{2} \underline{\sigma}(Q) - h , \\
    \mu_2 &= \frac{1}{2} c E_0 \underline{\sigma}(Q), \\ \quad  
    g  &=  -\frac{1}{2} \bigg [ \frac{\bar{\sigma}(P)\bar{\sigma}(A)}{\underline{\sigma}(D+B)}||\bigtriangleup||_F||\bar{\lambda}|| + \bar{\sigma}(P_1) \bigg] .
\end{align*}
\(V_z(z)\) is positive definite whenever the following conditions \cite{lewis2013} are met:
\begin{enumerate}
    \item \(K\) is positive definite.
    \item \(||z|| > \frac{||\omega||}{\underline{\sigma}(K)}\).
\end{enumerate}
For \textit{\textbf{condition 1}} to be satisfied, according to Sylvester's criterion, we must ensure 
\( 
    \beta > 0,
\) 
\( 
    \frac{\beta}{2}\kappa > 0,
\) 
\( 
    \frac{\beta}{2}\kappa \kappa_w > 0,
\) 
\( 
    \frac{\beta}{2}\kappa \kappa_w \kappa_0 > 0,
\) 
and \( 
    \frac{\beta}{2}\kappa \kappa_w (\kappa_0  \mu_1 - \gamma_3^2) - \frac{\beta}{2}\kappa \gamma_2^2 
    \kappa_0 - \frac{\beta}{2}\gamma_1^2 \kappa_w  \kappa_0 - g^2 \kappa \kappa_w \kappa_0 > 0.  
\)
where, the last condition, with the isolation of \( \mu_1 \), yields
\begin{equation*}
    \mu_1 >
    \frac{1}{\frac{\beta}{2}\kappa \kappa_w \kappa_0} \Bigg[ \frac{\beta}{2}\kappa \kappa_w\gamma_3^2 + \frac{\beta}{2}\kappa \gamma_2^2 \kappa_0  + \frac{\beta}{2}\gamma_1^2 \kappa_w  \kappa_0 +  g^2 \kappa \kappa_w \kappa_0 \Bigg]
\end{equation*}
For \textit{\textbf{condition 2}} to be satisfied, we require \( ||z|| > B_d \), where
\begin{equation}
    B_d = \frac{||\omega||_1}{\underline{\sigma}(K)}
\end{equation}
with \( \omega = \bigg [ 0,\kappa \Theta_n,\kappa_w \Theta_{nw}, \kappa_0 \Theta_{n0}, \Lambda \bigg]^{\top} \) and \( ||\omega||_1 = \kappa \Theta_n + \kappa_w \Theta_{nw} + \kappa_0 \Theta_{n0} + \Lambda  \). Thus,
\begin{equation}
    B_d = \frac{\kappa \Theta_n + \kappa_w \Theta_{nw} + \kappa_0 \Theta_{n0} + \Lambda}{\underline{\sigma}(K)}
\end{equation}
If \( ||z|| > B_d \), then \(\dot{V} \leq -V_z(z)\) with \(V_z(z)\) being positive definite. Since by \cite{lewis2013}, the Lyapunov function \eqref{allV} satisfies
\begin{equation} \label{zdef}
    \underline{\sigma}(\eta)||z||^2 \leq V \leq \bar{\sigma}(\zeta)||z||^2,
\end{equation}
then for any initial \( z(t_0) \), there exists \( T_0 \) such that
\begin{equation} \label{zref}
    ||z(t)|| \leq \sqrt{\frac{\bar{\sigma}(\zeta)}{\underline{\sigma}(\eta)}}B_d, \quad \forall t \geq t_0 + T_0,
\end{equation}
where 
\(\eta = \text{diag}(\frac{\underline{\sigma}(P)}{2}, \frac{1}{2\bar{\sigma}(F)} ,\frac{\underline{\sigma}(F_0)}{2}, \frac{1}{2\bar{\sigma}(F_w)}, \frac{\underline{\sigma}(P_1)}{2} ,\frac{1}{2\bar{\sigma}(Q)})\) and 
\( \zeta = \text{diag}(\frac{\bar{\sigma}(P)}{2}, \frac{1}{2\underline{\sigma}(F)} ,\frac{\bar{\sigma}(F_0)}{2}, \frac{1}{2\underline{\sigma}(F_w)}, \frac{\bar{\sigma}(P_1)}{2} ,\frac{1}{2\underline{\sigma}(Q)})\).

Let \( \upsilon = \text{min}_{{||z||}\ge B_d} V_z(z)\), then
\begin{equation}
    T_0 = \frac{V(t_0) - \bar{\sigma}(\zeta)(B_d)^2}{\upsilon}
\end{equation}
This implies \(r(t)\) is ultimately bounded, which yields \(e_i(t)\), \(e^n(t)\), and \(\delta^n\) CUUB, thus achieving synchronization over~\(\mathcal{G}\).

\textbf{Part 2:} To complete the proof of Theorem 1, we show that $x_i(t)$ is bounded for \( \forall t \ge t_0\) and \( \forall i \in \mathcal{N}\). From \eqref{modv}:

\begin{equation} \label{cond2}
    \dot{V} \le -\underline{\sigma}(K)||z||^2 + ||\omega||||z||.
\end{equation}
which, together with \(\eqref{zdef}\), we obtain
\begin{equation}
    \frac{d}{dt}(\sqrt{V}) \le - \frac{\underline{\sigma}(K)}{2\bar{\sigma}(\zeta)}\sqrt{V} + \frac{||\omega||}{2\sqrt{\underline{\sigma}(\eta)}}
\end{equation}
Thus, it follows from \cite[Corollary 1.1]{hut2006}, that \(V(t)\) remains bounded for all \(t \ge t_0\). This results, together with Assumption \ref{assp}, yields \(\delta_{i}^k = (x_{i}^k - \psi_{i}) - (x_{0}^k - \psi_{0})\) being~CUUB, i.e., \(||x_i|| \le X_n\) and \(||x_0|| \le X_{n0}\) are bounded~$\forall t \ge t_0$.
\end{proof}

\section{NUMERICAL EXAMPLE} \label{NUMERICALEXAMPLE}
We consider a leader-follower multi-agent formation with one leader and five followers, where each agent has distinct nonlinear dynamics influenced by external disturbances. Figure \ref{fig:1} shows the information flow structure of the multi-agent system. The leader’s dynamics are given by:

\textbf{Leader}:
\begin{align}
    \dot{s_0} &= v_0 \\
    \dot{v_0} &= - 3v_0 + 1  - g \sin(\alpha(s_0))  - \frac{0.4 v_0^2}{m_0} \nonumber \\
    & + \frac{3\sin(2t) + 6\cos(2t)}{m_0}  - \frac{(s_0 + v_0 - 1)^2}{3 m_0} (s_0 + 4v_0 - 1)  \notag
\end{align}
Each follower \(i\) (\(i = 1, \dots, 5\)) has second-order nonlinear dynamics:

\noindent \textbf{Agent 1}:
    \begin{align}
    \dot{s_1} &= v_1 \\
    \dot{v_1} &= \frac{v_1 \sin(s_1)}{m_1} + \cos^2(v_1) - \frac{0.47 v_1^2}{m_1}   - g \sin(\alpha(s_1))  
    + \frac{u_1}{m_1} \notag \\ &\quad + \frac{\zeta_1}{m_1} \notag
    \end{align}

\noindent  \textbf{Agent 2}:
    \begin{align}
    \dot{s_2} &= v_2 \\
    \dot{v_2} &= \frac{-(s_2)^2 v_2}{m_2} + \cos^2(v_2)  - \frac{0.52 v_2^2}{m_2}  - g  \sin(\alpha(s_2))   + \frac{u_2}{m_2} \notag \\ &\quad + \frac{\zeta_2}{m_2}\notag
    \end{align}

\noindent  \textbf{Agent 3}:
    \begin{align}
    \dot{s_3} &= v_3 \\
    \dot{v_3} &= \frac{-(s_3)^2 v_3}{m_3} + \sin^2(v_3) - \frac{0.57 v_3^2}{m_3}  - g \sin(\alpha(s_3))  + \frac{u_3}{m_3} \notag \\ &\quad + \frac{\zeta_3}{m_3}  \notag
    \end{align}

\noindent  \textbf{Agent 4}:
    \begin{align}
    \dot{s_4} &= v_4 \\
    \dot{v_4} &= \frac{-3(s_4 + v_4 - 1)^2 (s_4 + v_4 - 1)}{m_4} - v_4 + 0.5 \sin(2t)   \nonumber \notag \\
    &\quad  + \cos(2t)  - \frac{0.65 v_4^2}{m_4}  - g \sin(\alpha(s_4)) + \frac{u_4}{m_4} + \frac{\zeta_4}{m_4}  \notag
    \end{align}

\noindent  \textbf{Agent 5}:
    \begin{align}
    \dot{s_5} &= v_5 \\
    \dot{v_5} &= \cos(s_5) - \frac{0.74 v_5^2}{m_5}   - g \sin(\alpha(s_5)) + \frac{u_5}{m_5} + \frac{\zeta_5}{m_5} \notag
    \end{align}

\begin{table*}[h!]
    \centering
    \caption{Parameters for Leader and Followers}
    \resizebox{\textwidth}{!}{%
    \begin{tabular}{|c|c|c|c|c|c|c|}
        \hline
        \textbf{Parameter} & \textbf{Leader} ($i = 0$) & \textbf{Agent 1} ($i = 1$) & \textbf{Agent 2} ($i = 2$) & \textbf{Agent 3} ($i = 3$) & \textbf{Agent 4} ($i = 4$) & \textbf{Agent 5} ($i = 5$) \\ \hline
         $m_i$ & 2000 & 1200 & 1100 & 1500 & 1400 & 1500 \\ \hline
         $g$ & 9.81 & 9.81 & 9.81 & 9.81 & 9.81 & 9.81 \\ \hline
         $\alpha(s_i)$ & $0.05 \sin(0.1 s_0)$ & $0.05 \sin(0.1 s_1)$ & $0.05 \sin(0.1 s_2)$ & $0.05 \sin(0.1 s_3)$ & $0.05 \sin(0.1 s_4)$ & $0.05 \sin(0.1 s_5)$ \\ 
         \hline
         $\zeta_i$ & Time-Varying & 5 & 5 & 5 & 5 & 5 \\ \hline
    \end{tabular}%
    }
    \label{table:vehicle_parameters}
\end{table*}

\begin{figure}
    \centering
    \includegraphics[width=0.4\textwidth]{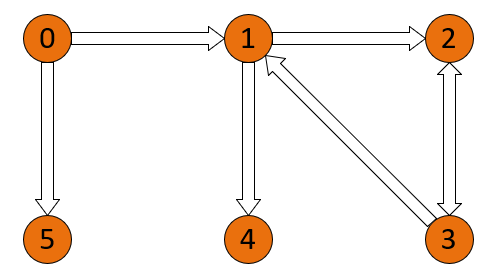}
    \caption{The Considered Fixed Topology of the augmented graph \(\bar{\mathcal{G}}\) in Section~\ref{NUMERICALEXAMPLE}.}
    \label{fig:1}
\end{figure}

\begin{figure}
    \centering
    \includegraphics[width=0.5\textwidth]{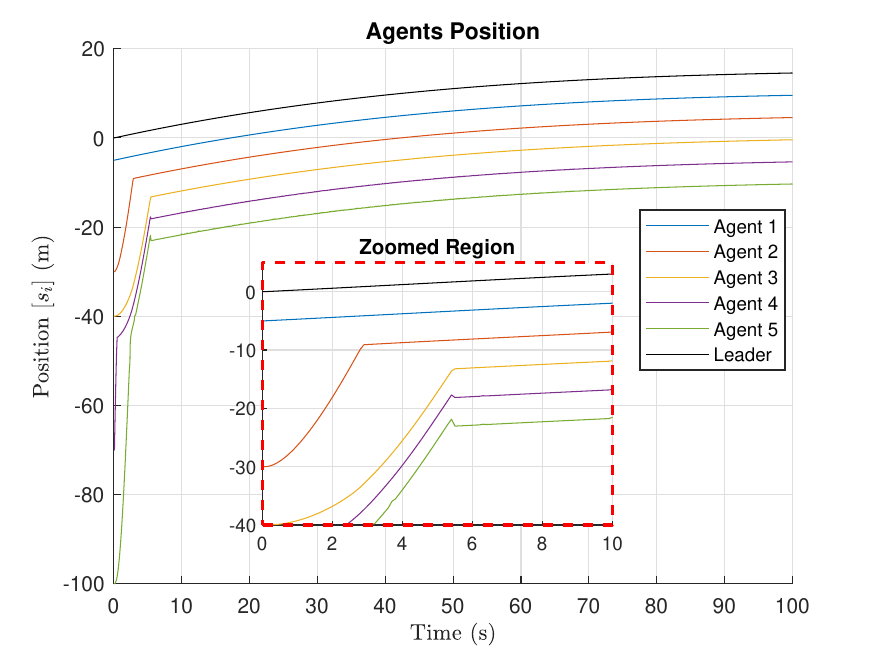}
    \caption{Agents Positions $s_i$, for $i \in \{1,2,3,4,5\}$ and the leader $0$.}
    \label{fig:2}
\end{figure}

\begin{figure}
    \centering
    \includegraphics[width=0.5\textwidth]{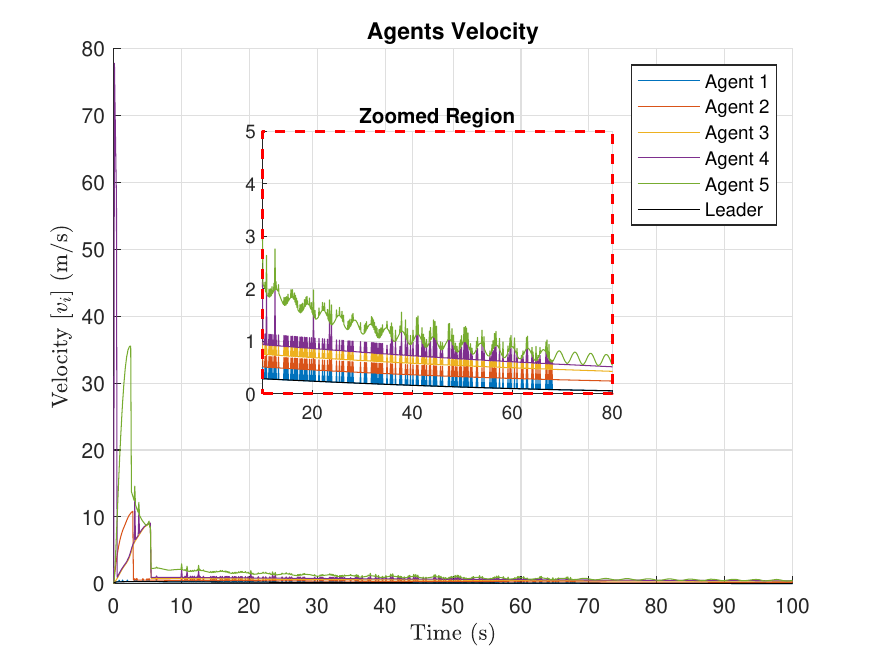}
    \caption{Agents Velocities $v_i$, for $i \in \{1,2,3,4,5\}$.}
    \label{fig:3}
\end{figure}

\begin{figure}
    \centering
    \includegraphics[width=0.5\textwidth]{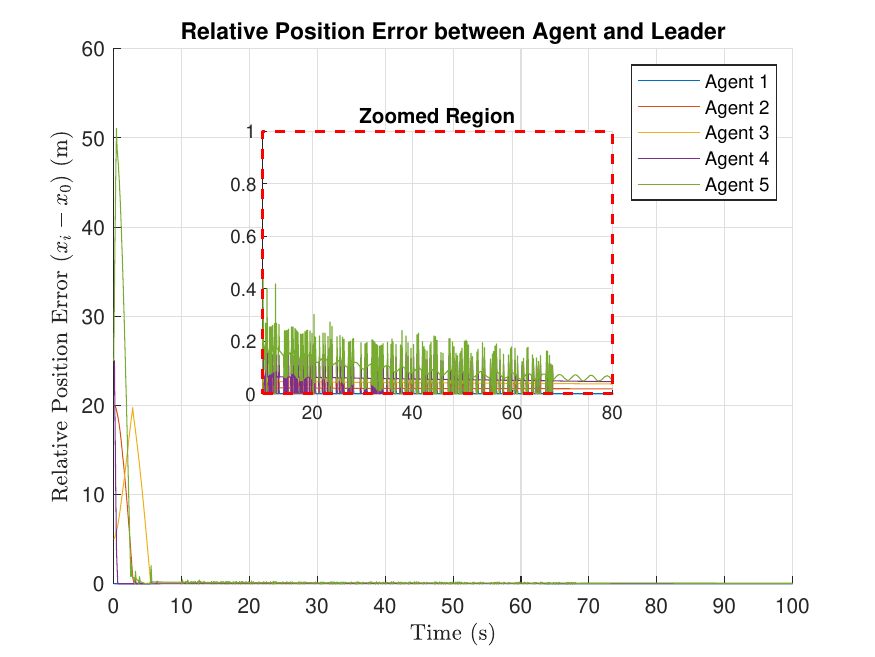}
    \caption{Relative Position Error, $E_{i0}$, for the first state, $k=1$, between Agents and Leader.}
    \label{fig:4}
\end{figure}

\begin{figure}
    \centering
    \includegraphics[width=0.5\textwidth]{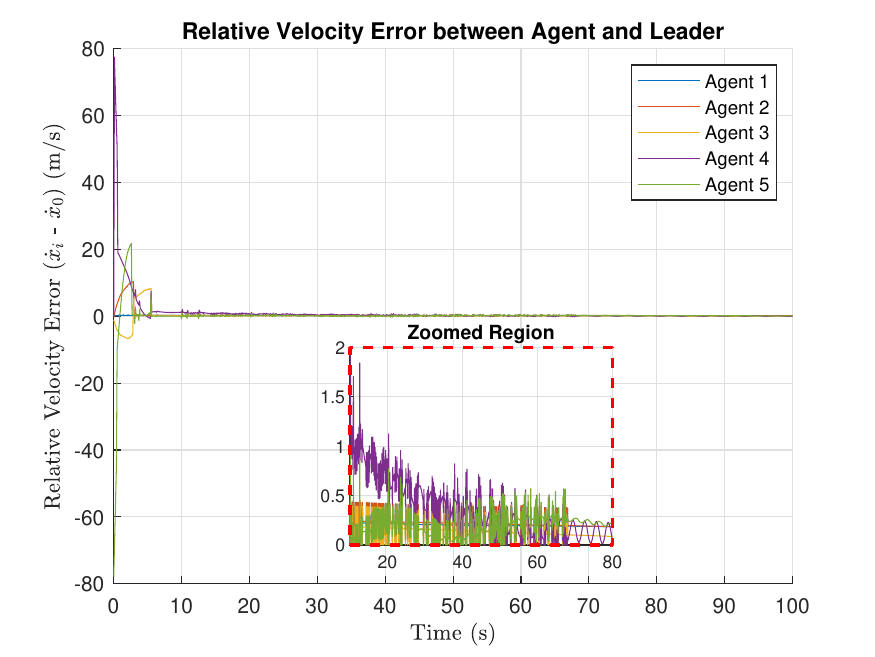}
    \caption{Relative Velocity Error between Agents and Leader.}
    \label{fig:5}
\end{figure}

\begin{figure}
    \centering
    \includegraphics[width=0.5\textwidth]{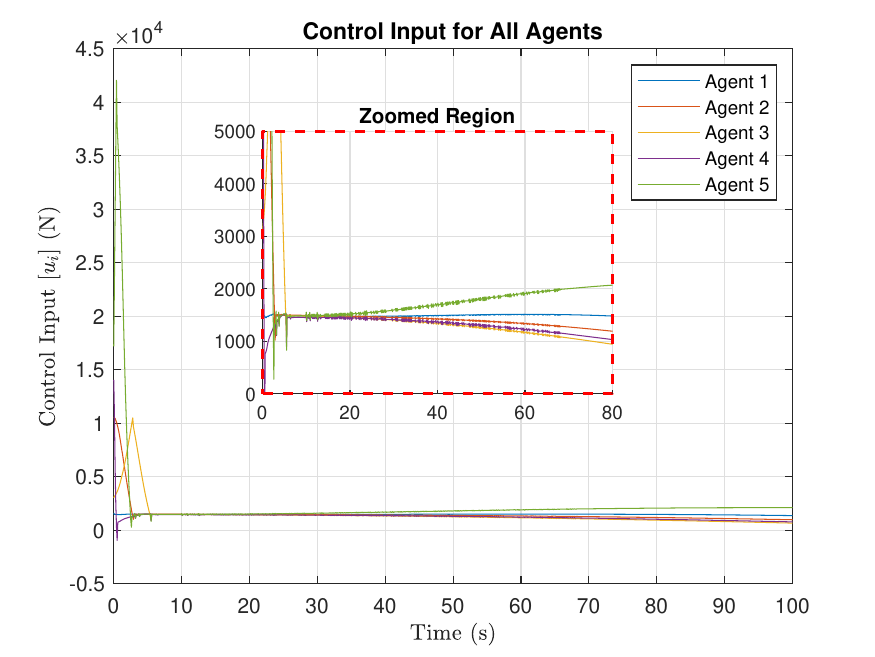}
    \caption{Agents Control Inputs $u_i$, for $i \in \{1,2,3,4,5\}$.}
    \label{fig:6}
\end{figure}

Figures \ref{fig:2} to \ref{fig:6} provide a detailed analysis of how the proposed control strategy ensures synchronization between the follower agents and the leader, while effectively maintaining system stability and handling disturbances. Initially, all agents start at different positions below the leader, demonstrating the system's capacity to regulate inter-agent distances and maintain a cohesive formation, as depicted in Figure \ref{fig:2}. As time progresses, the agents' velocities gradually converge towards the leader’s speed, reflecting the system’s ability to dynamically synchronize agents while compensating for initial trajectory discrepancies, as shown in Figure \ref{fig:3}. The relative position error between the agents and the leader, shown in Figure \ref{fig:4}, remains minimal throughout, suggesting that the control strategy maintains precise positioning and is robust against external disturbances. In Figure \ref{fig:5}, velocity errors exhibit an initial transient phase with notable variability, particularly for Agent 4, but ultimately converge, indicating the system's strong capacity for adaptation and stabilization over time. This highlights the control system’s resilience in managing fluctuations while ensuring long-term synchronization. Furthermore, the control inputs depicted in Figure \ref{fig:6} reveal distinct patterns across agents, reflecting varied demands based on their respective dynamic conditions. Agent 1, for example, requires a consistently higher control input, suggesting it may be compensating for greater dynamic challenges, while Agent 5’s input gradually increases, indicating evolving conditions that demand more effort to achieve synchronization. These observations demonstrate a well-coordinated control mechanism that balances both global synchronization and individual agent stability. The system ensures that all agents successfully converge toward a shared trajectory despite the initial discrepancies and disturbances, reflecting the overall effectiveness of the control strategy.

\section{CONCLUSION}\label{CONCLUSION}

This paper establishes a distributed adaptive control methodology for leader-follower formation consensus in heterogeneous multi-agent systems, incorporating collision and obstacle avoidance under fixed network topologies.

By employing neural network-based models to account for uncertainties in both the dynamics and disturbances in real-time, along with potential functions for collision and obstacle avoidance, the proposed Lyapunov-based control laws and adaptive tuning laws ensure synchronization, stability, and collision/obstacle avoidance. Numerical results illustrate the effectiveness of the proposed methodology, demonstrating synchronization and safety in a complex network of five agents with distinct nonlinear dynamics.

Future work includes extending the methodology to large networks of heterogeneous multi-agent systems with changing topologies and in the presence of communication delays.

\bibliographystyle{ieeetr}

\end{document}